\documentclass{amsproc}

\usepackage{feynarts}

\newtheorem{theorem}{Theorem}[section]
\newtheorem{lemma}[theorem]{Lemma}

\theoremstyle{definition}
\newtheorem{definition}[theorem]{Definition}

\theoremstyle{remark}

\numberwithin{equation}{section}

\begin{document}
\title{Multiple polylogarithms and linearly reducible \\ Feynman graphs}


\author{Christian Bogner and Martin L\"uders}
\address{Institute of Physics, Humboldt Universit\"at
zu Berlin\\
Unter den Linden 6, 10099 Berlin, Germany}
\curraddr{}
\email{bogner@math.hu-berlin.de, lueders@physik.hu-berlin.de}
\thanks{The first author is grateful to Francis Brown for useful communication and hospitality in Paris Jussieu, supported by ERC grant no. 257638. We thank Erik Panzer for useful communication and Dirk Kreimer's group at Humboldt University for hospitality and support. Our graphs were drawn using \cite{HahLan}.}

\subjclass[2000]{Primary }

\date{}

\begin{abstract}
We review an approach for the computation of Feynman integrals by
use of multiple polylogarithms, with an emphasis on the related criterion
of linear reducibility of the graph. We show that the set of graphs
which satisfies the linear reducibility with respect to both Symanzik
polynomials is closed under taking minors. As a step towards a classification                         
of Feynman integrals, we discuss the concept of critical minors and exhibit 
an example at three loops with four on-shell legs.

\end{abstract}
\maketitle


\section{Introduction}

In recent years we witnessed rapid progress in the developement of
techniques for the computation of higher order corrections in perturbative
quantum field theory. While other talks at this conference cover progress
in the computation of entire amplitudes, our talk refers to the
'classical' approach of computing the amplitude by its Feynman graphs,
which is inevitable when meeting the needs of present collider
experiments. In this field of research, it has shown to be fruitful
to discuss Feynman integrals in their own right, without restrictions
to a particular quantum field theory.

Computations of higher order corrections to observables often
start from the consideration of hundreds or thousands of Feynman integrals
with tensor structure, and proceed via effective standard procedures
to reduce the problem, possibly to a relatively small number of scalar
integrals. At higher loop-orders, the evaluation of the latter remains
to be the hard part of the problem. There is no algorithm which would
succeed in the analytical computation of every Feynman integral. However,
there is a variety of powerful methods which have been useful for
a wide range of relevant cases, such as the Mellin-Barnes approach (see \cite{BerLam, Uss, Tau, Smi}),
the expansion of hypergeometric functions \cite{MocUweWei, HubMai, HubMai2}, differential equation
methods \cite{Kot, Rem, GehRem, MueWeiZay, MueWeiZay2}, difference equations \cite{Lap, Tar, Tar2, Lee} or position-space
methods \cite{CheKatTka} (also see \cite{Gra}). In this talk we focus on the approach of iteratively integrating
out Feynman parameters by use of multiple polylogarithms.

In order to choose an appropriate strategy for the computation of
a given Feynman integral, it would be desirable in general, to know
in advance, which are the classes of functions and numbers the integral
may evaluate to. As a slightly more refined question of this type
we may ask: Which scalar Feynman integrals can be expressed by multiple
polylogarithms and multiple zeta values, and for which integrals do
we need a wider range of functions and numbers? In the past few years,
questions of this type turned out to define a fruitful common field
of research for quantum field theorists and algebraic geometers alike.
While the physicist's interest in these questions is given by the
desire to compute specific integrals or to learn about the 'number
content' of a given quantum field theory, the mathematician arrives
at the same question from a different direction. In a very general
context, Feynman integrals can be viewed as period integrals, and
the question of evaluating to multiple zeta values is related to the
question wether an underlying motive is mixed Tate over $\mathbb{Z}$ 
(see \cite{BloEsnKrei, Bro08, Bro09, BroSch, AluMar, AluMar2, AluMar3}).

A definite classification of Feynman graphs with respect to the above
questions is missing. However, for vacuum and two-point graphs important
progress was made by considering the first Symanzik polynomial, given
by the Feynman parametric representation of a Feynman integral. Even
though many vacuum-type Feynman integrals evaluate to multiple zeta
values \cite{BroaKre, BroaKre2}, this is not the case in general. A first vacuum graph whose
period has to belong to a set of numbers beyond multiple zeta values
was exhibited in a recent article by Brown and Schnetz \cite{BroSch} (also see \cite{BroDor}). When allowing
the Feynman integrals to depend on kinematical invariants and particle
masses, we can ask for graphs where multiple polylogarithms are not
sufficient to express the result. Here the first cases show up at
much lower loop-order, such as in the case of massive sunrise graphs
and related graphs with a cut through three massive edges (see e.g. \cite{Bauetal, Bauetal2}).

In this talk we review a criterion on graphs which is related to the
above questions and show that if a graph satisfies the criterion,
its minors do so as well. In graph theory such a minor monotony is an important and desireable feature. 
In section 2 we begin with a brief reminder
on scalar Feynman integrals, their two Symanzik polynomials and the
approach of integrating out Feynman parameters by use of multiple
polylogarithms. In section 3 we briefly review the criterion of linear
reducibility of a graph, which is used to decide whether a given integral
can be computed by use of the method. If this is the case, the functions
and numbers in all intermediate steps and in the result will not exceed
combinations of multiple polylogarithms and their values at rational points.
In this way the criterion and the corresponding algorithm are useful
tools for adressing the above questions. In the case of integrals
only involving the first Symanzik polynomial, the criterion was extensively
studied in \cite{Bro08, Bro09}. As the iterated integration over Feynman parameters
can be expected to be useful in the case of integrals depending on
kinematical invariants and particle masses as well, we intend to extend
the discussion to the second Symanzik polynomial. In section 4 we
consider linear reducibility with respect to both Symanzik polynomials
and show that the set of linearly reducible graphs is closed under
taking minors. This property is useful for a classification, as it
allows us to characterize families of reducible graphs by a small number of graphs not belonging to the family. In a case study we exhibit
such a 'forbidden minor' at the level of massless three-loop graphs
with four on-shell legs. Section 5 contains our conclusions.


\section{Multiple Polylogarithms and Feynman Integrals}

In this section we recall some general facts about Feynman integrals,
Symanzik polynomials and a method to compute period integrals by use
of multiple polylogarithms. Let us begin with a generic Feynman graph
$G$ with $n$ edges, loop-number (i.e. first Betti number) $L\geq1$
and with $r$ external half-edges (or 'legs'). We label each edge $e_{i}$
by an integration variable $\alpha_{i}$ (Feynman parameter), an integer
$\nu_{i}$ (exponent of the Feynman propagator), a real or complex
variable $m_{i}$ (particle mass). Each leg is labelled by a vector
$p_{j}$ (external momentum). 

To this labelled graph $G$ we associate the scalar Feynman integral in Dimensional Regularization:\begin{equation}
I_{G}=\frac{\Gamma(\nu-LD/2)}{\prod_{j=1}^{n}\Gamma(\nu_{j})}\int_{\alpha_{j}\geq0}\delta\left(1-\sum_{i=1}^{n}\alpha_{i}\right)\left(\prod_{j=1}^{n}d\alpha_{j}\alpha_{j}^{\nu_{j}-1}\right)\frac{\mathcal{U}_{G}^{\nu-(L+1)D/2}}{\mathcal{F}_{G}{}^{\nu-LD/2}}\label{eq:generic I_G}\end{equation}
where $\nu=\sum_{i=1}^{n}\nu_{i}.$ (We omit to write a trivial prefactor
by which the integral becomes independent of the physical mass-scale.)
The Feynman integral $I_{G}$ and the function $\mathcal{F}_{G}$
depend on the particle masses and on certain kinematical invariants,
which are quadratic functions of the external momenta. The functions
$\mathcal{U}_{G}$ and $\mathcal{F}_{G}$ are the first and second
Symanzik polynomial of the graph. A definition is given below. Usually
a Feynman integral is associated to a Feynman graph by Feynman rules
in momentum or position space, and we refer to the literature \cite{ItzZub, Nak} for
the standard computation leading from there to the Feynman parametric
representation given in eq. \ref{eq:generic I_G}. 

Eq. \ref{eq:generic I_G} defines a very general class of integrals
which deserves our attention for several reasons. Firstly, the class
contains the Feynman integrals of scalar quantum field theory such
as $\phi^{3}-$ or $\phi^{4}-$theory. Secondly, any Feynman integral
with a tensor-structure, arising from a physical quantum field theory,
can in principle be expressed in terms of scalar integrals of the
above class \cite{Tar, Tar2}. Thirdly, as we allow the $\nu_{j}$ to take arbitrary
integer values, there are well-known identities between these scalar
integrals which can be used for efficient reduction procedures \cite{CheTka}. As
a consequence, integrals of the above class appear in a wide range
of physical set-ups and their evaluation is the bottleneck of many
computational problems in particle physics.

The parameter $D$ can either be fixed to the integer space-time dimension
or, as the integral is very often ill-defined in the desired dimension,
one may consider $I_{G}$ in Dimensional Regularisation where
$D$ is a complex variable. Then, in order to separate the pole-terms
and obtain finite contributions in four-dimenional Minkowski space,
one usually attempts to compute the coefficients of a Laurent-expansion
\[
I_{G}=\sum_{j=j_{0}}^{\infty}c_{j}\epsilon^{j},\]
with $D=4-2\epsilon,$ to a desired order. Even though the computation
of the functions $c_{j}$ can be very difficult, we can make a general
statement about them. It is shown in \cite{BogWei09} that if for an arbitrary
Feynman graph we evaluate any function $c_{j}$ at algebraic values
of the squared particle masses $m_{i}^{2}$ and kinematical invariants
$s_{i}$, where all $m_{i}^{2}\geq0$ and all $s_{i}\leq0,$ we obtain
a \emph{period} according to the definition of Kontsevich and Zagier \cite{KonZag}.
For the special case where the Feynman integral takes the form 
\begin{equation}
\mathcal{P}_{G}=\int_{\alpha_{j}\geq0}\delta\left(1-\sum_{i=1}^{n}\alpha_{i}\right)\left(\prod_{j=1}^{n}d\alpha_{j}\right)\frac{1}{\mathcal{U}_{G}^{D/2}}\label{eq:period integral}\end{equation}
this statement was already proven in \cite{BelBro}. It seems that Feynman integrals
in fact evaluate to a restricted subset of periods and it is an important challenge to understand which one this is.

Let us now recall the definition of the Symanzik polynomials of a
Feynman graph $G.$ The first Symanzik polynomial is defined as\begin{eqnarray*}
\mathcal{U}_{G} & = & \sum_{T}\prod_{e_{i}\notin T}\alpha_{i},\end{eqnarray*}
where the sum is over all spanning trees of the graph $G.$ The second
Symanzik polynomial is defined as \[
\mathcal{F}_{G}=\mathcal{F}_{0,\, G}+\mathcal{U}_{G}\sum_{i=1}^{n}\alpha_{i}m_{i}^{2}\]
with \[
\mathcal{F}_{0,\, G}=\sum_{(T_{1},\, T_{2})}\left(\prod_{e_{i}\notin(T_{1},\, T_{2})}\alpha_{i}\right)s_{(T_{1},\, T_{2})}.\]
Here the sum runs through all spanning two-forests $(T_{1},\, T_{2})$
of $G$, where $T_{1}$ and $T_{2}$ denote the connected components of
the forest.

In order to define the kinematical invariants $s_{(T_{1},\, T_{2})},$
we introduce an arbitrary orientation on $G.$ We firstly say that
each external momentum $p_{j}$ is incoming at the vertex at the corresponding
leg. We furthermore label each oriented edge by a momentum-vector
$q_{i}$. If the edge $e_{i}$ is oriented from vertex $v_{j}$ to
$v_{k}$ then $q_{i}$ is said to be incoming at $v_{k}$ and $-q_{i}$
is incoming at $v_{j}.$ Momentum-conservation on $G$ is reflected
in our labels by the condition that the sum of all external momenta
$p_{j}$ is zero, and at each vertex, the sum of all incoming momenta
is zero. By these conditions, except for $L$ momenta, each of the
$q_{i}$ can be expressed as a linear combination of external momenta.
The kinematical invariants are defined as \[
s_{(T_{1},\, T_{2})}=\left(\sum_{e_{j}\notin(T_{1},\, T_{2})}\pm q_{j}\right)^{2}\]
where the sign of $q_{j}$ is fixed by the condition that we sum over
the momenta incoming at the component $T_{2}.$ Note that by momentum
conservation, the $s_{(T_{1},\, T_{2})}$ are functions of the external
momenta. 

As an alternative to the above construction by spanning trees and
forests, there are several ways to obtain both Symanzik polynomials from determinants of 
certain matrices \cite{BogWei10, BloKre, Pat, Bro09}. To demonstrate such a derivation, let us label each
edge $e_{i}$ by an auxiliary variable $y_{i}$. Each vertex $v_{i}$
is labelled by \[
u_{i}=\left\{ \begin{array}{c}
z_{j}\textrm{ if a leg with incoming momentum }p_{j}\textrm{ is attached,}\\
0\textrm{ if no leg is attached.}\end{array}\right.\]
 For a Feynman graph with vertices $v_{1},\,...,\, v_{m}$
we consider an $m\times m$ matrix $M$ whose entries are: \[
M_{ij}=\left\{ \begin{array}{c}
u_{i}+\sum y_{k}\textrm{ for }i=j,\, e_{k}\textrm{ attached to }v_{i}\textrm{ at exactly one end,}\\
-\sum y_{k}\textrm{ for }i\neq j,\, e_{k}\textrm{ connecting }v_{i}\textrm{ and }v_{j}.\end{array}\right.\]
We compute the determinant \[
\mathcal{V}(y_{1},\,..,\, y_{n},\, z_{1},\,...,\, z_{r})=\textrm{det}(M)\]
and consider the function \[
\mathcal{W}(\alpha_{1},\,..,\,\alpha_{n},\, z_{1},\,...,\, z_{r})=\mathcal{V}\left(\alpha_{1}^{-1},\,..,\,\alpha_{n}^{-1},\, z_{1},\,...,\, z_{r}\right)\prod_{i=1}^{n}\alpha_{i}\]
which is a polynomial in the $\alpha-$ and $z-$variables. Note that
$M$ depends on a chosen ordering on the vertices but $\mathcal{W}$
does not. 

Let us assume that at least two legs are attached to the graph, i.e.
$r\geq2$. We expand $\mathcal{W}$ as $\mathcal{W}=\mathcal{W}^{(1)}+\mathcal{W}^{(2)}+...+\mathcal{W}^{(r)}$
where $\mathcal{W}^{(k)}$ is homogeneous of degree $k$ in the $z-$variables.
We can directly read off the first Symanzik polynomial from the first
term in this expansion, as it satisfies \[
\mathcal{W}^{(1)}(\alpha_{1},\,..,\,\alpha_{n},\, z_{1},\,...,\, z_{r})=\mathcal{U}_{G}(\alpha_{1},\,..,\,\alpha_{n})\sum_{i=1}^{r}z_{i}.\]
The massless second Symanzik polynomial $\mathcal{F}_{0,\, G}$ is
directly obtained from $\mathcal{W}^{(2)}$. By construction, $\mathcal{W}^{(2)}$
is homogeneous of degree 2 in the $z-$variables. We replace each
product $z_{i}z_{j}$ in $\mathcal{W}^{(2)}$ by the scalar-product
of the corresponding external momentum vectors $p_{i}\cdot p_{j}$.
By momentum-conservation, $\sum_{i=1}^{r}p_{i}=0,$ we express each
of the scalar products by the functions $s_{(T_{1},\, T_{2})}$. As result we obtain $\mathcal{F}_{0,\, G}$ \cite{BogWei10}.%
\begin{figure}
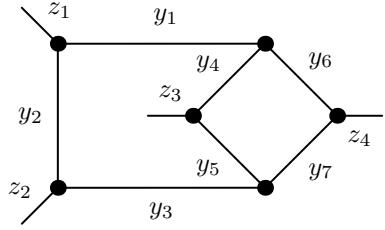

\begin{feynartspicture}(300,150)(1, 1)%
\FADiagram{} \FAProp(2.,13.)(2.,5.)(0.,){Straight}{0} \FALabel(1.18,9.)[r]{$y_2$} \FAProp(2.,5.)(13.5,5.)(0.,){Straight}{0} \FALabel(7.75,4.18)[t]{$y_3$} \FAProp(2.,13.)(13.5,13.)(0.,){Straight}{0} \FALabel(8.,14.02)[b]{$y_1$} \FAProp(13.5,13.)(9.5,9.)(0.,){Straight}{0} \FALabel(11.0606,11.4392)[br]{$y_4$} \FAProp(13.5,13.)(17.5,9.)(0.,){Straight}{0} \FALabel(15.9034,11.6243)[bl]{$y_6$} \FAProp(9.5,9.)(13.5,5.)(0.,){Straight}{0} \FALabel(11.0607,6.5606)[tr]{$y_5$} \FAProp(17.5,9.)(13.5,5.)(0.,){Straight}{0} \FALabel(15.9034,6.3756)[tl]{$y_7$} \FAProp(9.5,9.)(7.,9.)(0.,){Straight}{0} \FALabel(8.25,9.82)[b]{$z_3$} \FAProp(17.5,9.)(20.,9.)(0.,){Straight}{0} \FALabel(18.75,8.18)[t]{$z_4$} \FAProp(2.,13.)(0.,15.)(0.,){Straight}{0} \FALabel(1.4392,14.4392)[bl]{$z_1$} \FAProp(2.,5.)(-0.,3.)(0.,){Straight}{0} \FALabel(0.5607,4.4392)[br]{$z_2$} \FAVert(2.,13.){0} \FAVert(2.,5.){0} \FAVert(13.5,5.){0} \FAVert(13.5,13.){0} \FAVert(9.5,9.){0} \FAVert(17.5,9.){0}
\end{feynartspicture}%

\caption{The non-planar double-box\label{fig:The-non-planar-double-box}}

\end{figure}

As an example let us compute the two Symanzik polynomials of the massless
non-planar double-box, shown in figure \ref{fig:The-non-planar-double-box}
with auxiliary $y-$ and $z-$variables. For this graph and a chosen
ordering on the vertices we have \[
M=\left(\begin{array}{cccccc}
M_{11} & -y_{2} & 0 & 0 & -y_{1} & 0\\
-y_{2} & M_{22} & -y_{3} & 0 & 0 & 0\\
0 & -y_{3} & M_{33} & -y_{7} & 0 & -y_{5}\\
0 & 0 & -y_{7} & M_{44} & -y_{6} & 0\\
-y_{1} & 0 & 0 & -y_{6} & M_{55} & -y_{4}\\
0 & 0 & -y_{5} & 0 & -y_{4} & M_{66}\end{array}\right),\]
where $M_{11}=y_{1}+y_{2}+z_{1},$ $M_{22}=y_{2}+y_{3}+z_{2},$ $M_{33}=y_{3}+y_{5}+y_{7},$
$M_{44}=y_{6}+y_{7}+z_{4},$ $M_{55}=y_{1}+y_{4}+y_{6},$ $M_{66}=y_{4}+y_{5}+z_{3}.$
Proceeding in the described way we compute \begin{eqnarray*}
\mathcal{U}_{G} & = & (z_{1}+z_{2}+z_{3}+z_{4})^{-1}\mathcal{W}^{(1)}\\
 & = & (\alpha_{1}+\alpha_{2}+\alpha_{3})(\alpha_{4}+\alpha_{5}+\alpha_{6}+\alpha_{7})+(\alpha_{4}+\alpha_{5})(\alpha_{6}+\alpha_{7}),\\
\mathcal{F}_{0\, G} & = & \mathcal{W}^{(2)}|_{z_{i}z_{j}=p_{i}\cdot p_{j},\,\sum_{i=1}^{4}p_{i}=0}\\
 & = & -p_{1}^{2}\alpha_{2}(\alpha_{1}(\alpha_{4}+\alpha_{5}+\alpha_{6}+\alpha_{7})+\alpha_{4}\alpha_{6})\\
 &  & -p_{2}^{2}\alpha_{2}(\alpha_{3}(\alpha_{4}+\alpha_{5}+\alpha_{6}+\alpha_{7})+\alpha_{5}\alpha_{7})\\
 &  & -p_{3}^{2}(\alpha_{4}\alpha_{5}(\alpha_{1}+\alpha_{2}+\alpha_{3}+\alpha_{6}+\alpha_{7})+\alpha_{3}\alpha_{4}\alpha_{7}+\alpha_{1}\alpha_{5}\alpha_{6})\\
 &  & -p_{4}^{2}(\alpha_{6}\alpha_{7}(\alpha_{1}+\alpha_{2}+\alpha_{3}+\alpha_{4}+\alpha_{5})+\alpha_{1}\alpha_{4}\alpha_{7}+\alpha_{3}\alpha_{5}\alpha_{6})\\
 &  & -(p_{1}+p_{2})^{2}(\alpha_{1}\alpha_{3}(\alpha_{4}+\alpha_{5}+\alpha_{6}+\alpha_{7})+\alpha_{1}\alpha_{5}\alpha_{7}+\alpha_{3}\alpha_{4}\alpha_{6})\\
 &  & -(p_{1}+p_{3})^{2}\alpha_{2}\alpha_{5}\alpha_{6}-(p_{2}+p_{3})^{2}\alpha_{2}\alpha_{4}\alpha_{7}.\end{eqnarray*}

It is often sufficient to consider the Feynman integral after setting
some of its legs on-shell, which means that the corresponding external
momenta are fixed by setting their square to a squared particle mass.
In our example we may assume massless particles and set $p_{i}^{2}=0$ for all $i=1,\,...,\,4$. The corresponding
Feynman integral was evaluated in dimensional regularization by classical
polylogarithms in reference \cite{Tau}. We will return to Symanzik polynomials of graphs
with four on-shell legs in section \ref{sec:Classification-by-Critical}. 

Let us now turn to iterated integrals. Let $k$ be the field of either
the real or the complex numbers and $M$ a smooth manifold over $k$.
We consider a piecewise smooth path on $M$, given by a map $\gamma:[0,\,1]\rightarrow M$,
and some smooth differential 1-forms $\omega_{1},\,...,\,\omega_{n}$
on $M.$ The iterated integral of these 1-forms along the path $\gamma$
is defined by \[
\int_{\gamma}\omega_{n}...\omega_{1}=\int_{0\leq t_{1}\leq...\leq t_{n}\leq1}f_{n}(t_{n})dt_{n}...f_{1}(t_{1})dt_{1},\]
where $f_{i}(t)dt=\gamma^{\star}(\omega_{i})$ is the pull-back of
$\omega_{i}$ to $[0,\,1].$ With the term iterated integral we will
more generally refer to $k$-linear combinations of such integrals.

We will consider classes of iterated integrals which define the same
function for any two homotopic paths. Such integrals are called \emph{homotopy
invariant}. They are well-defined functions of variables given by
the end-point of $\gamma$. In such iterated integrals the differential
forms and the order in which we integrate over them have to satisfy
a property known as the \emph{integrability condition}. The condition
is best formulated on tensor products of 1-forms over some field $K\subseteq k$,
which we denote by $\left[\omega_{1}|...|\omega_{m}\right].$ Let
$D$ denote a $K$-linear map from tensor products of smooth 1-forms
on $M$ to tensor products of all forms on $M,$ given by \begin{eqnarray*}
D\left(\left[\omega_{1}|...|\omega_{m}\right]\right) & = & \sum_{i=1}^{m}\left[\omega_{1}|...|\omega_{i-1}|d\omega_{i}|\omega_{i+1}|...|\omega_{m}\right]+\sum_{i=1}^{m-1}\left[\omega_{1}|...|\omega_{i-1}|\omega_{i}\wedge\omega_{i+1}|...|\omega_{m}\right].\end{eqnarray*}
A $K$-linear combination of tensor products $\xi=\sum_{l=0}^{m}c_{i_{1},\,...,\, i_{l}}\left[\omega_{i_{1}}|...|\omega_{i_{l}}\right]$,
$c_{i_{1},\,...,\, i_{l}}\in K$, is called an \emph{integrable word}
if it satisfies the equation \[
D\xi=0.\]
Let $\Omega$ be a finite set of smooth 1-forms and let $B_{m}(\Omega)$
denote the vector space of integrable words of length $m$ with 1-forms
in $\Omega.$ Now we return from words to integrals by considering
the integration map on integrable words:

\begin{equation}
\sum_{l=0}^{m} \sum_{i_{1},\,...,\,i_{l}}c_{i_{1},\,...,\, i_{l}}\left[\omega_{i_{1}}|...|\omega_{i_{l}}\right]\mapsto\sum_{l=0}^{m}\sum_{i_{1},\,...,\,i_{l}}c_{i_{1},\,...,\, i_{l}}\int_{\gamma}\omega_{i_{1}}...\omega_{i_{l}}.\label{eq:Integration map}\end{equation}

A fundamental theorem of Chen \cite{Che} states that this map is an isomorphism
from $B_{m}(\Omega)$ to the set of homotopy invariant iterated integrals
in 1-forms in $\Omega$ of length less or equal to $m,$ if $\Omega$ satisfies
further conditions which we do not specify here. 

In the following we fix $K=\mathbb{Q}$ and discuss two sets of 1-forms
for which the theorem applies. For a coordinate $t_{1}$ on an open
subset of $\mathbb{C}$ we firstly consider the set of closed 1-forms
\[
\Omega_{n}^{\textrm{Hyp}}=\left\{ \frac{dt_{1}}{t_{1}},\,\frac{dt_{1}}{t_{1}-1},\,\frac{t_{2} dt_{1}}{t_{1}t_{2}-1},\,...,\,\frac{\left(\prod_{i=2}^{n}t_{i}\right)dt_{1}}{\prod_{i=1}^{n}t_{i}-1}\right\} .\]

As a trivial consequence of $dt_{1}\wedge dt_{1}=0$, any
tensor product of 1-forms in $\Omega_{n}^{\textrm{Hyp}}$ is an integrable
word. By applying the integration map eq. \ref{eq:Integration map}
to these words, we obtain the class of hyperlogarithms \cite{Lapp}. In particle
physics it is very common to use sub-classes of hyperlogarithms. As
an example, we may consider $\Omega_{2}^{\textrm{Hyp}}$ and fix the
constant $t_{2}=-1.$ To physicists, the iterated integrals obtained from this restriction
are well known as harmonic polylogarithms \cite{RemVer} and suffice for the evaluation
of many Feynman integrals. 

We want to focus on a class of functions of severable variables, obtained
from another set of 1-forms, where now all the $t_{1},\,...,\, t_{n}$
are considered to be coordinates in an open subset of $\mathbb{C}^{n}:$\[
\Omega_{n}^{\textrm{MPL}}=\left\{ \frac{dt_{1}}{t_{1}},\,...,\,\frac{dt_{n}}{t_{n}},\frac{d\left(\prod_{a\leq i\leq b}t_{i}\right)}{\prod_{a\leq i\leq b}t_{i}-1}\textrm{ where }1\leq a\leq b\leq n\right\} .\]

For this set the integrability condition is not trivial and there
are words for which it is not satisfied. The homotopy invariant iterated
integrals which we obtain via the integration map from the the integrable
words in $\Omega_{n}^{\textrm{MPL}}$ form the vector space $\mathcal{B}\left(\Omega_{n}\right)$
of multiple polylogarithms in $n$ variables. We use the notation
$\mathcal{B}_{m}\left(\Omega_{n}\right)$ for the vector space of
such functions obtained from integrable words of length $\leq m.$
There is an explicit map \cite{BogBro} to construct all integrable words in $\Omega_{n}^{\textrm{MPL}}$,
closely related to the 'symbol' in \cite{Gon09, Gon10, Duh11}.

The functions in $\mathcal{B}\left(\Omega_{n}\right)$ were extensively
studied in reference \cite{Bro06}. We just want to recall
a few statements which are relevant for the following considerations.
Firstly, the multiple polylogarithms of Goncharov \cite{Gon01},
frequently used in the physics literature, are contained in this class.
As we want to use the elements of $\mathcal{B}\left(\Omega_{n}\right)$
in an iterative integration procedure, it is important for us to know
their primitives and limits. It is proven in \cite{Bro06} that $\mathcal{B}\left(\Omega_{n}\right)$
is closed under taking primitives. Furthermore if we take the limits
of elements of $\mathcal{B}\left(\Omega_{n}\right)$ at $t_{n}$ equal
to 0 and 1, we obtain $\mathcal{Z}$-linear combinations of elements
in $\mathcal{B}\left(\Omega_{n-1}\right),$ where $\mathcal{Z}$ denotes
the $\mathbb{Q}$-vector space of multiple zeta values.

Now let us consider definite integrals of the form 
\begin{equation} I=\int_{0}^{1}dt_{n}\frac{\beta(\{g_{i}\})}{f}  \label{eq:int form} \end{equation} 
where $f$ is a polynomial and $\beta(\{g_{i}\})\in\mathcal{B}\left(\Omega_{n}\right)$
is a multiple polylogarithm whose arguments are some irreducible polynomials
$g_{i}$. Let us call $f$ and the $g_{i}$ the \emph{critical polynomials}
of the integrand. If $f$ and the $g_{i}$ are linear in $t_{n}$
we can evaluate the above integral and from the mentioned properties
it is clear that the result will be a $\mathcal{Z}$-linear combination
of elements in $\mathcal{B}_{m}\left(\Omega_{n-1}\right).$ If the
result can be again expressed by functions of the form of the above
generic integrand and the critical polynomials are linear in $t_{n-1}$
then we can continue and integrate over this variable from 0 to 1,
and so on. 

Such an iterative procedure can be used to compute Feynman integrals. For recent examples in the  
physics literature, partly relying on different parametrizations, we refer to \cite{DelDuhGloSmi, DelDuhSmi, AblBluHasKleSchWis, ChaDuh, AnaDuhDulMis}. Aiming at such a computation one has to express the Feynman integral by a finite parametric
integral such that the integrand can be written in the above form,
where after each integration step, the critical polynomials are linear
in at least one of the remaining parameters. The method was introduced
systematically in \cite{Bro08} and demonstrated for certain Feynman parametric integrals of the type of eq. \ref{eq:period integral},
coming from primitive logarithmically divergent vacuum Feynman graphs.
However, the approach is not restricted to such graphs. Reference \cite{BroKre}
presents a method to express Feynman integrals with UV sub-divergences
by finite parametric integrals to which the approach may apply. The
treatment of graphs with infrared divergences is not excluded in principle,
but we are missing a canonical method to express IR-divergent integrals by finite
ones. In principle, the method of sector decomposition \cite{BinHei} allows us to write down the coefficients
of a Laurent expansion for a dimensionally regularized, infrared divergent
integral in terms of finite integrals over Feynman parameters, however, 
the polynomials in these integrals usually become very complicated. In view of the above approach one would ideally wish for a method, 
where the critical polynomials in the finite integrals
could be obtained from the Symanzik polynomials in a rather simple way.

For the following discussion let us assume, that in some way we have already been able
to express a given Feynman integral by finite integrals of the type
$I$ and that the critical polynomials are the Symanzik polynomials
of the graph. We focus on the criterion, that after each integration
over a Feynman parameter, the new critical polynomials have to be
linear in a next Feynman parameter. The reduction algorithm to be
reviewed in section \ref{sec:Linear-Reducibility} allows us to study
this criterion as it computes for each integration step a set in which
the critical polynomials are contained. %
\begin{figure}
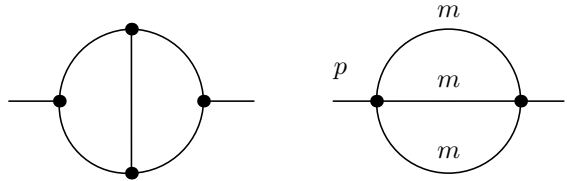

\begin{feynartspicture}(250,120)(2, 1)%

\FADiagram{}
\FAProp(1.5,10.)(5.,10.)(0.,){Straight}{0} \FALabel(3.25,9.18)[t]{} \FAProp(5.,10.)(10.,15.)(-0.4,){Straight}{0} \FALabel(6.0607,13.9392)[br]{} \FAProp(10.,15.)(15.,10.)(-0.4,){Straight}{0} \FALabel(13.9392,13.9392)[bl]{} \FAProp(15.,10.)(18.5,10.)(0.,){Straight}{0} \FALabel(16.75,9.18)[t]{} \FAProp(5.,10.)(10.,5.)(0.4527,){Straight}{0} \FALabel(5.9288,5.9288)[tr]{} \FAProp(10.,5.)(15.,10.)(0.4,){Straight}{0} \FALabel(13.9392,6.0607)[tl]{} \FAProp(10.,15.)(10.,5.)(0.,){Straight}{0} \FALabel(9.18,10.)[r]{} \FAVert(5.,10.){0} \FAVert(10.,15.){0} \FAVert(15.,10.){0} \FAVert(10.,5.){0} 

\FADiagram{}
\FAProp(2.,10.)(5.,10.)(0.,){Straight}{0} \FALabel(2.5,11.52)[b]{$p$} \FAProp(5.,10.)(15.,10.)(-1.,){Straight}{0} \FALabel(10.,15.82)[b]{$m$} \FAProp(15.,10.)(18.,10.)(0.,){Straight}{0} \FALabel(16.5,9.18)[t]{} \FAProp(5.,10.)(15.,10.)(1.,){Straight}{0} \FALabel(10.,6.02)[b]{$m$} \FAProp(5.,10.)(15.,10.)(0.,){Straight}{0} \FALabel(10.,11.02)[b]{$m$} \FAVert(5.,10.){0} \FAVert(15.,10.){0}  
\end{feynartspicture}%

\caption{(a) Massless two-loop graph, (b) Equal-mass sunrise graph\label{fig:two-loop prop graphs}}

\end{figure}

As a further motivation of the following discussion, let us have a glance at two
well-known Feynman graphs in view of the mentioned criterion. For
the massless two-loop graph of figure \ref{fig:two-loop prop graphs}
(a) it was proven by use of the Mellin-Barnes approach and expansions
by nested sums that each coefficient of the $\epsilon$-expansion
is a combination of multiple zeta values \cite{BieWei}. Reference \cite{Bro08}
confirmed this statement for this two-loop graph and several higher-loop graphs by relating them to integrals
of the type of eq. \ref{eq:period integral} whose integrands satisfy
the criterion.

The case of the equal-mass two-loop sunrise graph,
shown in figure \ref{fig:two-loop prop graphs} (b), is very different.
The desired coefficients in the $\epsilon$-expansion can be derived
from the $D=2$-dimensional version of the Feynman integral, \[
I_{\textrm{sunrise}}=\int_{\alpha_{j}\geq0}d\alpha_{1}d\alpha_{2}d\alpha_{3}\delta\left(1-\sum_{i=1}^{3}\alpha_{i}\right)\frac{1}{\mathcal{F}_{G}},\]
with the second Symanzik polynomial\[
\mathcal{F}_{G}=-p^{2}\alpha_{1}\alpha_{2}\alpha_{3}+m^{2}(\alpha_{1}\alpha_{2}+\alpha_{2}\alpha_{3}+\alpha_{1}\alpha_{3})(\alpha_{1}+\alpha_{2}+\alpha_{3}),\]
playing the role of the critical polynomial. As $\mathcal{F}_{G}$
is not linear in any of the Feynman parameters, integral $I_{\textrm{sunrise}}$
fails the criterion. It is in perfect match with this simple observation, that the known
result of the sunrise integral involves elliptic integrals \cite{LapRem}. In general, it is possible that the polynomials of an integrand fail the criterion and 
still the integral can be expressed by multiple polylogarithms. However,
the criterion may provide a useful first classification and give a
hint where to look for integrals, which exceed the class of multiple
polylogarithms. We also want to mention reference \cite{CarLar}, where different criteria are used to search for such integrals. 


\section{Linear Reducibility\label{sec:Linear-Reducibility}}

Let us briefly review the polynomial reduction algorithm of \cite{Bro08}.
Let $S=\{f_{1},\,...,\, f_{N}\}$ be a set of polynomials in the variables
$\alpha_{1},\,...,\,\alpha_{n}$ with rational coefficients. 

(1) If there is an index $1\leq r_{1}\leq n$ such that all polynomials
in $S$ are linear in $\alpha_{r_{1}}$ we can write\[
f_{i}=g_{i}\alpha_{r_{1}}+h_{i}\textrm{ for all }1\leq i\leq N,\]
where $g_{i}=\frac{\partial f_{i}}{\partial\alpha_{r_{1}}}$ and $h_{i}=f_{i}|_{\alpha_{r_{1}}=0}.$
We define\begin{equation}
S'_{(r_{1})}=\left\{ \left(g_{i}\right)_{1\leq i\leq N},\,\left(h_{i}\right)_{1\leq i\leq N},\,\left(h_{i}g_{j}-g_{i}h_{j}\right)_{1\leq i<j\leq N}\right\} \label{eq:S tilde}\end{equation}
and furthermore we define $S_{(r_{1})}$ to be the set of irreducible
polynomials in $S'_{(r_{1})}.$ In $S_{(r_{1})}$ we neglect all constants
and monomials.

(2) If there is a $1\leq r_{2}\leq n$ such that all polynomials in
$S_{(r_{1})}$ are linear in $\alpha_{r_{2}},$ we repeat the above
step, now with $S_{(r_{1})}$ and $\alpha_{r_{2}}$ in the roles of
$S$ and $\alpha_{r_{1}},$ and obtain a new set of polynomials which
we call $S_{[r_{1}](r_{2})}.$ Then, assuming that starting from $S$
the above steps can be done first for $\alpha_{r_{2}}$ and then for
$\alpha_{r_{1}},$ we compute $S_{[r_{2}](r_{1})}$ and take the intersection
of both sets:\[
S_{[r_{1},\, r_{2}]}=S_{[r_{1}](r_{2})}\cap S_{[r_{2}](r_{1})}.\]
Whenever we speak of intersections here and in the following, we mean the common zero loci,
such that if a polynomial appears in two sets with a different constant prefactor, 
it nevertheless belongs to the intersection. Then we choose a next variable in which all polynomials of $S_{[r_{1},\, r_{2}]}$
are linear and continue in the same way. At each iteration we apply
step (1) and take the intersection\[
S_{[r_{1},\, r_{2},\,...,\, r_{k}]}=\cap_{1\leq i\leq k}S_{[r_{1},\,...,\,\hat{r_{i}},\,...,\, r_{k}](r_{i})}.\]
If a set $S_{[r_{1},\,...,\,\hat{r_{i}},\,...,\, r_{k}]}$ contains
a polynomial which is non-linear in $\alpha_{r_{i}}$, the set
 $S_{[r_{1},\,...,\,\hat{r_{i}},\,...,\, r_{k}](r_{i})}$ is
undefined and omitted in the intersection. If this happens for all
$1\leq i\leq k$ the set $S_{[r_{1},\, r_{2},\,...,\, r_{k}]}$ is
undefined and the algorithm stops. Unless this situation occurs, we
obtain for each sequence of variables $(\alpha_{r_{1}},\,...,\,\alpha_{r_{k}})$,
$k\leq n,$ a sequence of sets $S_{(r_{1})},\, S_{[r_{1},\, r_{2}]},\,...,\, S_{[r_{1},\, r_{2},\,...,\, r_{k}]}.$ 
\begin{definition}
We say that $S$ is \emph{Fubini reducible} (or \emph{linearly reducible})
if there is an ordering of all $n$ variables $(\alpha_{r_{1}},\,...,\,\alpha_{r_{n}})$
such that every polynomial in $S_{[r_{1},\, r_{2},\,...,\, r_{k}]}$
is linear in $\alpha_{r_{k+1}}$ for all $1\leq k<n.$
\end{definition}

The linear reducibility of the set of critical polynomials of an integrand
as in eq.\ \ref{eq:int form} is a criterion for the integral to be
computable by the above approach. The criterion is sufficient but not 
neccessary. The sets $S_{[r_{1},\,...,\, r_{k}]}$ contain the critical polynomials
of the integrand after the first $k$ integrations, but might as well contain spurious 
polynomials which drop out in the integration procedure. A more refined reduction algorithm
presented in reference \cite{Bro09} omits such cases, but not necessarily all of them. Furthermore the occurrence of a quadratic polynomial
does not always forbid us to continue with the computation.

By applying the above algorithm to first Symanzik polynomials, it was shown
in \cite{Bro08} that several vacuum Feynman integrals can be computed
and evaluate to combinations of multiple zeta values. Moreover the
same is true for coefficients of a dimensional series expansion of
Feynman integrals with two legs, obtained from cutting one edge in
one of these vacuum graphs, as in the case of figure \ref{fig:two-loop prop graphs} (a). 
The linear reducibility of first Symanzik polynomials is extensively studied
in terms of Dodgson polynomials in \cite{Bro09}. It is shown that the
first five iterations of the reduction succeed for any first Symanzik
polynomial. Moreover, a first Symanzik polynomial is reducible, if
its graph has vertex width less or equal three. (This is the class
of graphs which decompose into two connected components after removing
three vertices or less.) These results explain, why one has to go
up to complicated graphs at high loop orders to find first examples
for vacuum-type Feynman integrals which exceed the set of multiple
zeta values \cite{BroSch, BroDor}.

In the following, we want to consider the above algorithm applied
to both Symanzik polynomials. In order to include $\mathcal{F}$ which
may depend on particle masses and external momenta, we slightly extend
the above formulation of step (1), allowing $f_{i}$ to be polynomials
whose coefficients are rational numbers or algebraic functions of
additional parameters $s_{1},\,...,\, s_{m}.$ The rest of the algorithm
is not affected by this change.

It will be useful to consider
polynomial reduction in  coordinates for products of $\mathbb{P}^1$. Let $P(\alpha_{1},\,...,\,\alpha_{n})$
be a polynomial in $n$ Feynman parameters and consider an associated
function of new variables $x_{1},\,...,\, x_{n},\, y_{1},\,...,\, y_{n}$
defined by \[
\bar{P}(x_{1},\,...,\, x_{n},\, y_{1},\,...,\, y_{n})=\left(\prod_{i=1}^{n}y_{i}^{d_{i}(P)}\right)P\left(\frac{x_{1}}{y_{1}},\,...,\,\frac{x_{n}}{y_{n}}\right),\]
where $d_{i}(P)$ denotes the maximal degree of $P$ in $\alpha_{i}$
for all $1\leq i\leq n.$ We obtain our original polynomial back by
setting \begin{equation}
\bar{P}(\alpha_{1},\,...,\,\alpha_{n},\,1,\,...,\,1)=P(\alpha_{1},\,...,\,\alpha_{n}).\label{eq:back to alpha}\end{equation}
The function $\bar{P}$ is by definition a polynomial in $x_{1},\,...,\, x_{n},\, y_{1},\,...,\, y_{n}$
and it is linear in $x_{k}$ if and only if $P$ is linear in $\alpha_{k}.$

Let us consider a set of polynomials $S=\{P_{1},\,...,\, P_{N}\}$
which are all linear in $\alpha_{k}.$ We can write \[
P_{i}=\alpha_{k}\frac{\partial}{\partial\alpha_{k}}P_{i}+P_{i}|_{\alpha_{k}=0}.\]
Consider the corresponding set of polynomials $\bar{S}=\{\bar{P}_{1},\,...,\,\bar{P}_{N}\},$
obtained from $S$ by changing to the projective coordinates. Each
of these polynomials is linear in $x_{k}$ and we can write\[
\bar{P}_{i}=x_{k}\bar{P}_{i}|_{x_{k}=1,\, y_{k}=0}+y_{k}\bar{P}_{i}|_{x_{k}=0,\, y_{k}=1}.\]
Let us use the convention that in a Fubini reduction with respect
to the $x$-variables, we neglect the above prefactor $y_{k}$ in
the sense that in eq. \ref{eq:S tilde} the terms are given by $g_{i}=\bar{P}_{i}|_{x_{k}=1,\, y_{k}=0}$
and $h_{i}=\bar{P}_{i}|_{x_{k}=0,\, y_{k}=1},$ which does not affect
the linear reducibility. Step (1) of the reduction algorithm applied
to $\bar{S}$ with respect to $x_{k}$ then gives the set $\bar{S}_{(k)}$
consisting of the irreducible factors of 

\begin{eqnarray}
\bar{S}'_{(k)} & = & \left\{ \left(\bar{P}_{i}|_{x_{k}=1,\, y_{k}=0}\right)_{1\leq i\leq N},\,\left(\bar{P}_{i}|_{x_{k}=0,\, y_{k}=1}\right)_{1\leq i\leq N},\right.\nonumber \\
 &  & \left.\left(\bar{P}_{i}|_{x_{k}=0,\, y_{k}=1}\cdot\bar{P}_{j}|_{x_{k}=1,\, y_{k}=0}-\bar{P}_{i}|_{x_{1}=1,\, y_{k}=0}\cdot\bar{P}_{j}|_{x_{k}=0,\, y_{k}=1}\right)_{1\leq i<j\leq N}\right\} .\label{eq:step one auf bar S}\end{eqnarray}

Considering a Fubini reduction in $x$-variables instead of $\alpha$-variables,
we convince ourselves that the factorizations into irreducible polynomials
are not affected by the change of variables. Indeed if $P$ factorizes
as \[
P=f_{1}\cdot f_{2}\]
into two polynomials $f_{1},$ $f_{2}$ then $\bar{P}$ factorizes
as \begin{eqnarray*}
\bar{f_{1}}\cdot\bar{f_{2}} & = & \left(\prod_{i=1}^{n}y_{i}^{d_{i}(f_{1})}\right)f_{1}\left(\frac{x_{1}}{y_{1}},\,...,\,\frac{x_{n}}{y_{n}}\right)\cdot\left(\prod_{i=1}^{n}y_{i}^{d_{i}(f_{2})}\right)f_{2}\left(\frac{x_{1}}{y_{1}},\,...,\,\frac{x_{n}}{y_{n}}\right)=\bar{P},\end{eqnarray*}
because of $d_{i}(f_{1})+d_{i}(f_{2})=d_{i}(P)$ for all $1\leq i\leq n.$
As a consequence we obtain:
\begin{lemma}
$P$ is linearly reducible with respect to the variables $\alpha_{1},\,...,\,\alpha_{n}$
if and only if $\bar{P}$ is linearly reducible with respect to the
variables $x_{1},\,...,\, x_{n}.$
\end{lemma}
Now let us choose some coordinate $x_{l}$, $1\leq l\leq n,$ and
define two new sets of polynomials by restrictions $x_{l}=0,\, y_{l}=1$
and $x_{l}=1,\, y_{l}=0$ respectively: \begin{eqnarray*}
\bar{S}^{(l,\,0,\,1)} & = & \bar{S}|_{x_{l}=0,\, y_{l}=1},\\
\bar{S}^{(l,\,1,\,0)} & = & \bar{S}|_{x_{l}=1,\, y_{l}=0}.\end{eqnarray*}

To these sets we apply one reduction step with respect to $x_{k}$. 
At first let us assume $l\neq k.$ Step (1) of the algorithm
gives \begin{eqnarray*}
\bar{S}_{(k)}^{(l,\,0,\,1)} & = & \textrm{irreducible factors of }\bar{S}'_{(k)}|_{x_{l}=0,\, y_{l}=1},\\
\bar{S}_{(k)}^{(l,\,1,\,0)} & = & \textrm{irreducible factors of }\bar{S}'_{(k)}|_{x_{l}=1,\, y_{l}=0}.\end{eqnarray*}
If a polynomial $\bar{P}$ factorizes as $\bar{P}=f_{1}\cdot f_{2}$
then furthermore \begin{eqnarray*}
\bar{P}|_{x_{l}=0,\, y_{l}=1} & = & f_{1}|_{x_{l}=0,\, y_{l}=1}\cdot f_{2}|_{x_{l}=0,\, y_{l}=1},\\
\bar{P}|_{x_{l}=1,\, y_{l}=0} & = & f_{1}|_{x_{l}=1,\, y_{l}=0}\cdot f_{2}|_{x_{l}=1,\, y_{l}=0},\end{eqnarray*}
and the maximal degrees with respect to another variable $x_{i}$
satisfy $d_{i}(f_{j})\geq d_{i}\left(f_{j}|_{x_{l}=0,\, y_{l}=1}\right)$
and $d_{i}(f_{j})\geq d_{i}\left(f_{j}|_{x_{l}=1,\, y_{l}=0}\right)$
for $j=1,\,2$ and for all $i=1,\,...,\, n$. This means that for
each polynomial $f\in\bar{S}_{(k)}$ there is at most one polynomial
$f'\in\bar{S}_{(k)}^{(l,\,0,\,1)}$, and $d_{i}(f)\geq d_{i}(f')$
for all $i=1,\,...,\, n$. The same is true for $\bar{S}_{(k)}^{(l,\,1,\,0)}.$

Lastly we observe that in the case of $l=k,$ the irreducible factors
of $\bar{S}^{(l,\,0,\,1)}$ and $\bar{S}^{(l,\,1,\,0)}$ are already
contained in $\bar{S}_{(k)}$. This proves the following lemma:
\begin{lemma}
\emph{If $\bar{S}$ is linearly reducible with the ordering $\left(x_{r_{1}},\,...,\, x_{r_{n}}\right)$
then for any $1\leq l\leq n$ the sets }\textup{\emph{$\bar{S}|_{x_{l}=0,\, y_{l}=1}$
}}\textup{and }$\bar{S}|_{x_{l}=1,\, y_{l}=0}$ \emph{are linearly
reducible with $\left(x_{r_{1}},\,...,\,\hat{x}_{l},\,...,\, x_{r_{n}}\right).$ }
\end{lemma}
In combination with the previous lemma we obtain:
\begin{lemma}
Let $S=\{P_{1},\,...,\, P_{N}\}$ be a set of polynomials which is
linearly reducible with $(\alpha_{r_{1}},\,...,\,\alpha_{r_{n}})$
and whose members are linear in $\alpha_{l}$. Then the sets $S^{l}=\left\{ \frac{\partial P_{1}}{\partial\alpha_{l}},\,...,\,\frac{\partial P_{N}}{\partial\alpha_{l}}\right\} $
and $S_{l}=\left\{ P_{1}|_{\alpha_{l}=0},\,...,\, P_{N}|_{\alpha_{l}=0}\right\} $
are linearly reducible with $(\alpha_{r_{1}},\,...,\,\hat{\alpha}_{l},\,...,\,\alpha_{r_{n}})$.\label{lemma:Lin red Ableitung und Restriktion}
\end{lemma}


\section{Towards a Classification by Critical Minors\label{sec:Classification-by-Critical}}

Let $G$ be a graph and $E_{G}$ its set of edges. For $e\in G$
we denote by $G\backslash e$ the graph obtained from $G$ by \emph{deletion}
of $e.$ Furthermore we write $G//e$ for the graph obtained from
$G$ by \emph{contraction} of $e$. This is the graph where the end-points
of $e$ are identified and then $e$ is removed. For any distinct
edges $e_{1},\, e_{2}\in E_{G}$ the operations of deleting (or contracting)
$e_{1}$ and deleting (or contracting) $e_{2}$ commute. Therefore
we can more generally write $\gamma=G\backslash D//C$ with distinct
$D,\, C\subset E_{G}$, for the unique graph obtained from $G$ by
deleting all edges in $D$ and contracting all edges in $C.$ Any
such $\gamma$ is called a \emph{minor} of $G.$

If $G$ is connected and there is an edge $e$ such that $G\backslash e$
is disconnected then $e$ is called a \emph{bridge}. When speaking
of Feynman graphs, we may ignore disconnected graphs, so we introduce
the convention that the corresponding Symanzik polynomial are zero:\[
\mathcal{U}_{G\backslash e}=0,\,\mathcal{F}_{0,\, G\backslash e}=0\textrm{ if }e\textrm{ is a bridge.}\]
Let us furthermore call $e_{t}\in E_{G}$ a \emph{tadpole} if it is
attached to the same vertex at both ends. In this case we have a factorization
in the corresponding Feynman parameter:\begin{eqnarray*}
\mathcal{U}_{G} & = & \mathcal{U}_{G\backslash e_{t}}\alpha_{t},\\
\mathcal{F}_{0,\, G} & = & \mathcal{F}_{0,\, G\backslash e_{t}}\alpha_{t}.\end{eqnarray*}
For any $e$ which is not a tadpole we have the well-known deletion/contraction
identities: \begin{eqnarray*}
\mathcal{U}_{G} & = & \mathcal{U}_{G\backslash e}\alpha_{e}+\mathcal{U}_{G//e},\\
\mathcal{F}_{0,\, G} & = & \mathcal{F}_{0,\, G\backslash e}\alpha_{e}+\mathcal{F}_{0,\, G//e}.\end{eqnarray*}
Now we consider a set of arbitrary particle masses $m_{1},\,...,\, m_{n}$
distributed over the edges of $G,$ with the restriction that at least
one edge $e_{k}$ is massless, $m_{k}=0.$ Then the second Symanzik
polynomial\begin{eqnarray*}
\mathcal{F}_{G} & = & \mathcal{F}_{0,\, G}+\mathcal{U}_{G}\cdot\sum_{i\neq k}\alpha_{i}m_{i}^{2}\end{eqnarray*}
is linear in $\alpha_{k}$ and the above relations are extended by\begin{eqnarray*}
\mathcal{F}_{G} & = & \mathcal{F}_{G\backslash e_{k}}\alpha_{k}\textrm{ if }e_{k}\textrm{ is a tadpole,}\\
\mathcal{F}_{G} & = & \mathcal{F}_{G\backslash e_{k}}\alpha_{k}+\mathcal{F}_{G//e_{k}}\textrm{ if }e_{k}\textrm{ is not a tadpole.}\end{eqnarray*}

Equivalently, if $e_{k}$ is a massless edge of $G$ and not a tadpole then\begin{equation}
\frac{\partial}{\partial\alpha_{k}}\mathcal{U}_{G}=\mathcal{U}_{G\backslash e_{k}},\,\mathcal{U}_{G}|_{\alpha_{k}=0}=\mathcal{U}_{G//e_{k}},\label{eq:del contr 1}\end{equation}
\begin{equation}
\frac{\partial}{\partial\alpha_{k}}\mathcal{F}_{G}=\mathcal{F}_{G\backslash e_{k}},\,\mathcal{F}_{G}|_{\alpha_{k}=0}=\mathcal{F}_{G//e_{k}}.\label{eq:del contr 2}\end{equation}
In the following let us call a Feynman graph $G$ linearly reducible,
if $\{\mathcal{U}_{G},\,\mathcal{F}_{G}\}$ is linearly reducible.
We arrive at the main statement of these notes:
\begin{theorem}
If $G$ is a linearly reducible Feynman graph then any minor of $G$
is linearly reducible as well.\label{theorem:Minor closed}\end{theorem}
\begin{proof}
Let $G$ be an arbitrary Feynman graph, $e_{j}$ any of its edges
and $\tilde{G}=G|_{m_{j}=0}$ the graph obtained from $G$ by setting
the mass $m_{j}$ associated to $e_{j}$ equal to zero. Assume that
$G$ is linearly reducible. Then $G$ is linearly reducible for any
value of $m_{j}$ and therefore $\tilde{G}$ is linearly reducible
as well. As $e_{j}$ is massless in $\tilde{G}$ we can apply equations
\ref{eq:del contr 1} and \ref{eq:del contr 2} and obtain, if $e_{j}$ is not a tadpole, \begin{eqnarray*}
S_{\tilde{G}\backslash e_{j}} & = & \left\{ \mathcal{U}_{\tilde{G}\backslash e_{j}},\,\mathcal{F}_{\tilde{G}\backslash e_{j}}\right\} =\left\{ \frac{\partial}{\partial\alpha_{j}}\mathcal{U}_{\tilde{G}},\,\frac{\partial}{\partial\alpha_{j}}\mathcal{F}_{\tilde{G}}\right\} ,\\
S_{\tilde{G}//e_{j}} & = & \left\{ \mathcal{U}_{G//e_{j}},\,\mathcal{F}_{G//e_{j}}\right\} = \left\{ \mathcal{U}_{\tilde{G}}|_{\alpha_{j}=0},\,\mathcal{F}_{\tilde{G}}|_{\alpha_{j}=0}\right\}.
 \end{eqnarray*}
If $e_{j}$ is a tadpole we have 
\[ S_{\tilde{G}//e_{j}}=S_{\tilde{G}\backslash e_{j}}= \left\{ \frac{\partial}{\partial\alpha_{j}}\mathcal{U}_{\tilde{G}},\,\frac{\partial}{\partial\alpha_{j}}\mathcal{F}_{\tilde{G}}\right\}.\]
By lemma \ref{lemma:Lin red Ableitung und Restriktion} it follows from
the linear reducibility of $\tilde{G}$ that $\tilde{G}\backslash e_{j}$
and $\tilde{G}//e_{j}$ are linearly reducible. 
$G$ and $\tilde{G}$ have the same minors with respect to deleting
or contracting $e_{j}:$\[
G\backslash e_{j}=\tilde{G}\backslash e_{j},\, G//e_{j}=\tilde{G}//e_{j}\]
and therefore $S_{G\backslash e_{j}}=S_{\tilde{G}\backslash e_{j}}$
and $S_{G//e_{j}}=S_{\tilde{G}//e_{j}}.$ By induction, this proves
the theorem.
\end{proof}
A set of graphs $\mathcal{G}$ is called \emph{minor closed} if for
all $G\in\mathcal{G}$ every minor of $G$ belongs to $\mathcal{G}$
as well. Let $\mathcal{H}$ be any set of graphs and let $\mathcal{G}_{\mathcal{H}}$
be the set of all graphs which do not have a minor in the set $\mathcal{H}.$
Then the set $\mathcal{G}_{\mathcal{H}}$ is minor closed, and the
graphs in $\mathcal{H}$ are called \emph{forbidden minors} of $\mathcal{G}_{\mathcal{H}}.$
A theorem of Robertson and Seymour \cite{RobSey} states that any minor closed set
of graphs can be defined in such a way by a finite set of forbidden
minors. 

A graph $H$ is called a \emph{critical minor} of a set of graphs $\mathcal{G}$
if the minors of $H$ belong to $\mathcal{G}$ but $H$ does not. 
By removing a graph from $\mathcal{H}$ which is a minor of another graph in $\mathcal{H}$, we do not change the set $\mathcal{G}_{\mathcal{H}}$.
Therefore, to define a minor closed set $\mathcal{G}_{\mathcal{H}}$, it is sufficient to let $\mathcal{H}$ consist only of critical minors. 
A well-known example for a characterization by critical minors is
given by the set of planar graphs. Due to a theorem of Wagner \cite{Wag}, the
set of planar graphs is the set with forbidden minors $\{K_{3,3},\, K_{5}\}.$ 

In general we have to be careful when adapting notions from pure graph
theory to the study of Feynman graphs, which are equipped with labels
and special properties. However, theorem \ref{theorem:Minor closed} suggests to
attempt an analogous characterization of the set of linearly reducible
Feynman graphs by critical minors, i. e. by a set of Feynman graphs
which are not linearly reducible but have linearly reducible minors.
We conclude these notes with a first step towards such a characterization
for Feynman graphs with a dependence on external momenta. Let $\Lambda$ be the set of massless Feynman graphs with four on-shell
legs, attached to four distinct vertices. The on-shell condition restricts
the four external momenta to satisfy $p_{i}^{2}=0$ such that the
dependence on the external momenta can be expressed by two Mandelstam
variables. By use of a Maple-implementation of the Fubini algorithm, we find linear reductions for 
all two-loop graphs of this class. At three loops there are several graphs in this class for which our program fails
to find a linear reduction. Most of these graphs have the graph of figure \ref{fig:K4} as a minor. This is the complete four-vertex graph 
$K_{4}$ where we attached an on-shell leg at each vertex. From direct observation of the two Symanzik polynomials one can see that the 
graph is not linearly reducible. As on the other hand all its minors are linearly reducible, the graph plays the role of a critical minor for 
the set of all linearly reducible graphs. Our case study suggests that only a few further critical minors are 
needed to distinguish all linearly reducible graphs of $\Lambda$ at three loops. In this way, a large class of graphs can be separated into linearly reducible
and irreducible members by knowing only a small number of critical minors.

\begin{figure}
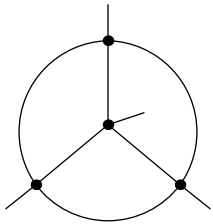

\begin{feynartspicture}(100,100)(1, 1)%

\FADiagram{}
\FAProp(10.,16.)(4.,4.)(0.6,){Straight}{0} \FALabel(2.7172,11.9013)[br]{} \FAProp(4.,4.)(16.,4.)(0.5,){Straight}{0} \FALabel(10.,0.1799)[t]{} \FAProp(16.,4.)(10.,16.)(0.6,){Straight}{0} \FALabel(17.2827,11.9013)[bl]{} \FAProp(10.,16.)(10.,9.)(0.,){Straight}{0} \FALabel(9.18,12.5)[r]{} \FAProp(4.,4.)(10.,9.)(0.,){Straight}{0} \FALabel(7.3522,5.9813)[tl]{} \FAProp(10.,9.)(16.,4.)(0.,){Straight}{0} \FALabel(12.6477,5.9813)[tr]{} \FAProp(10.,9.)(13.,10.)(0.,){Straight}{0} \FALabel(11.911,8.7467)[t]{} \FAProp(10.,16.)(10.,19.)(0.,){Straight}{0} \FALabel(10.82,17.5)[l]{} \FAProp(16.,4.)(18.5,2.)(0.,){Straight}{0} \FALabel(16.9178,2.4648)[tr]{} \FAProp(4.,4.)(1.5,2.)(0.,){Straight}{0} \FALabel(2.4178,3.5351)[br]{} \FAVert(10.,16.){0} \FAVert(4.,4.){0} \FAVert(16.,4.){0} \FAVert(10.,9.){0} 

\end{feynartspicture}%

\caption{$K_{4}$ with four legs \label{fig:K4}}

\end{figure}


\section{Conclusions}

In this talk we reviewed the criterion of linear reducibility of Symanzik polynomials, which can be used to decide whether 
a corresponding Feynman integral can be computed by iteratively introducing hyperlogarithms or multiple polylogarithms. 
Brown showed in \cite{Bro09} that with respect to the first Symanzik polynomial, the set of linearly reducible graphs is closed
 under taking minors and that therefore a characterization of this set by forbidden, critical minors is possible. We extended this line of argument 
to the case of graphs with masses and kinematical invariants, by showing, that minor closedness is true for both Symanzik polynomials. We exhibit a first 
critical minor with a non-trivial dependence on kinematical invariants. 

We expect a classification with respect to the criterion of linear reducibility to be useful for the more difficult question of which Feynman integrals
evaluate to multiple polylogarithms. Due to the simplicity of its derivation and the property of minor closedness, such a systematic classification 
of a large class of Feynman graphs appears feasible.

\bibliographystyle{amsplain}

\end{document}